\documentclass[12pt]{article}

\usepackage[margin=2.5cm]{geometry}

\usepackage{amsmath,amsthm,amsfonts,amscd,amssymb,bbm,mathrsfs,enumerate,url}

\usepackage{graphicx,tikz}
\usepackage[pdfborder={0 0 0}]{hyperref}
\usetikzlibrary{matrix,arrows}
\usetikzlibrary{decorations.pathreplacing}
\usetikzlibrary{decorations.pathmorphing}
\usetikzlibrary{patterns,fadings}

\newcounter{axioms}

\def\RR{{\mathbb R}}
\def\CC{{\mathbb C}}
\def\NN{{\mathbb N}}
\def\ZZ{{\mathbb Z}}

\def\A{{\mathcal A}}
\def\B{{\mathcal B}}

\def\H{{\mathcal H}}
\def\I{{\mathcal I}}

\def\M{{\mathcal M}}

\def\a{\alpha}
\def\b{\beta}
\def\d{\delta}

\def\f{\varphi}
\def\g{\gamma}

\def\k{\kappa}

\def\r{\rho}
\def\s{\sigma}

\def\Ad{{\hbox{\rm Ad\,}}}

\def\1{{\mathbbm 1}}
\def\Exp{{\rm Exp}}

\def\uone{{\rm U(1)}}

\def\diff{{\rm Diff}}

\def\diffs1{\diff(S^1)}

\def\vir{{\rm Vir}}
\def\Vir{{\rm Vir}}
\def\supp{{\rm supp\,}}

\def\psl2r{{\rm PSL}(2,\RR)}
\def\sl2r{{\rm SL}(2,\RR)}
\def\su11{{\rm SU}(1,1)}
\def\2dmob{{\overline{\psl2r}\times\overline{\psl2r}}}
\def\<{\langle}
\def\>{\rangle}

\def\im{\mathrm{Im}\,}

\def\hfin{\H^{\mathrm{fin}}}

\def\dom{{\mathrm{Dom}}}

\newtheorem{theorem}{Theorem}[section]

\newtheorem{proposition}[theorem]{Proposition}
\newtheorem{lemma}[theorem]{Lemma}
\theoremstyle{remark}

\title{Ground state representations of some non-rational conformal nets}
\date{}
 \author{
 {\bf Yoh Tanimoto}\footnote{Supported by Programma per giovani ricercatori, anno 2014 ``Rita Levi Montalcini''
of the Italian Ministry of Education, University and Research.}\\
 Dipartimento di Matematica, Universit\`a di Roma Tor Vergata\\
 Via della Ricerca Scientifica, 1, I-00133 Roma, Italy\\
 email: {\tt hoyt@mat.uniroma2.it}\\
 }
\begin{document}
\maketitle

\begin{abstract}
 We construct families of ground state representations of
 the $\mathrm{U}(1)$-current net and of the Virasoro nets $\vir_c$
 with central charge $c \ge 1$. We show that these representations are 
 not covariant with respect to the original dilations,
 and those on the $\uone$-current net are not solitonic.
 Furthermore, by going to the dual net with respect to the ground state representations
 of $\vir_c$, one obtains possibly new family of M\"obius covariant nets on $S^1$.
 \end{abstract}

\section{Introduction}\label{Introduction}
A model of quantum field theory may be in various states:
in quantum theory, a system is defined by the algebra of observables and
a physical state is realized as a normalized positive linear functional on it \cite{Haag96}.
Among them, the most important state on a quantum field theory is the vacuum state,
and indeed, notable axiomatizations of quantum field theory are formulated in terms of the vacuum state
or the vacuum correlation functions. Yet, some other states are of particular interest.
For example, DHR (Doplicher-Haag-Roberts, \cite{DHR69}) representations correspond to charged states.
KMS (Kubo-Martin-Schwinger, \cite[Section 5.3]{BR2}) states represent thermally equilibrium states.
Limiting cases of KMS states, where the temperature is zero, are called
ground states and they enjoy particular properties.
In this paper, we construct ground states on chiral components of certain two-dimensional conformal field theories.

Ground states can be characterized in various ways, see e.g.\! \cite{BR2}.
In one of them, ground states are invariant under the dynamics
and the generator of the dynamics has positive generator.
In the operator-algebraic setting, this is particularly interesting, because
if one has a ground state on a one-dimensional Haag-Kastler net, one can take the dual net in the GNS
(Gelfand-Naimark-Segal, \cite[Section 2.3.3]{BR1}) representation,
The dual net should be considered not as an extension but as a possibly new, different model
from the original net, because it has a different M\"obius symmetry.
Let us recall that there is a certain relation between the failure of Haag duality and
symmetry breaking (see \cite[Section 1]{BDLR92} for this and some implications on Goldstone bosons),
but this does not necessarily hold in $(1+1)$-dimensions,
because a Haag-Kastler net is not always the fixed point net of a larger net with respect to
a Lie group.

In this work, we focus on the simplest class of quantum field theory:
chiral components of two-dimensional conformal field theory.
Previously we studied KMS states on them \cite{CLTW12-1, CLTW12-2}\footnote{Let us stress
that we study here KMS states and ground states with respect to the spacetime translations.
We studied the KMS states with respect to rotations in a previous work \cite{LT18}
and discussed possible applications to 3d black holes.}.
Differently from KMS states, it appears to be difficult to classify ground states
for a given net: the techniques to classify KMS states do not directly apply
to ground states, because we do not have a direct connection between ground states and
modular automorphisms, which was crucial in \cite[Section 4.2.1, Appendix B]{CLTW12-1}\footnote{
The arguments for the case of maximal nets \cite[Theorem 4.7]{CLTW12-1} might still work for ground states.},
nor extension results
of KMS states to larger net (cf.\! \cite{AHKT77}\cite[Appendix A]{CLTW12-2}) are available for ground states.
On the other hand, we classified ground states on loop algebras \cite{Tanimoto11},
yet it is still unclear whether the results passes to ground states on the corresponding Haag-Kastler nets.
Therefore, rather than classifying them, we try to construct various examples.

Contrary to the case of loop algebras where ground state is unique \cite{Tanimoto11},
we find that the $\mathrm{U}(1)$-current net and the Virasoro nets with $c \ge 1$ admit continuously many ground states. 
It turns out that the $\mathrm{U}(1)$-current net admits a one-parameter family of automorphisms
commuting with translations, and by composing them with the vacuum state, we obtain
a family of ground states. As the Virasoro nets with $c \ge 1$ restricted to $\RR$ can be embedded
in the $\mathrm{U}(1)$-current net \cite{BS90}, they also admit a family of ground states.
In this way, we construct a family of pure ground states on the $\mathrm{U}(1)$-current net
parametrized by $q \in \RR$, and those on the Virasoro nets with $c \ge 1$ parametrized
by $\frac{q^2}2 \ge 0$.
In the GNS representations of these ground states, one can take the dual net.
While the dual nets in the case of the ground states on the $\mathrm{U}(1)$-current net
are unitarily equivalent to the $\mathrm{U}(1)$-current net itself,
for the case of the Virasoro nets with $c > 1$
the dual nets must be different from the original net, since the latter are not strongly additive.

We also obtain explicit expressions of the current and the stress-energy tensor in the GNS representation.
These expressions in turn serve to show that the GNS representations of the $\uone$-current net
are not normal on half-lines.
This shows that the implication ``positivity of energy $\Longrightarrow$ solitonic representation'',
conjectured in \cite[Conjecture 34]{Henriques17} for loop groups, does not hold for the $\mathrm{U}(1)$-current net.

This paper is organized as follows.
In Section \ref{preliminaries} we recall M\"obius covariant nets, operator-algebraic setting
of chiral components of conformal field theory and collect general facts on ground states.
In Section \ref{examples} we introduce examples of M\"obius covariant nets, the $\uone$-current net and the Virasoro nets.
In Section \ref{groundstates} we construct ground states on these nets and study their property.
We conclude with some open problems in Section \ref{outlook}.

\section{Preliminaries}\label{preliminaries}
\subsection{M\"obius covariant net on circle}\label{mob}
The following is a mathematical setting for chiral components of two-dimensional conformal field theory.
These chiral components are essentially quantum field theories on the real line $\RR$, and by conformal covariance
they extend to $S^1$. In the operator-algebraic setting, they are realized as
nets (precosheaves) of von Neumann algebras on $S^1$. More precisely, let $\I$ be the set of open, non-dense, non empty intervals in $S^1$.
A {\bf M\"obius covariant net} $(\A, U, \Omega)$ is a triple of
the map $\A: \I \ni I \mapsto \A(I) \subset \B(\H)$, where $\A(I)$ is a von Neumann algebra on a Hilbert space $\H$ and $\B(\H)$ is the set of all bounded operators on $\H$,
$U$ a strongly continuous representation of the group $\psl2r$ on $\H$
and a unit vector $\Omega \in \H$ satisfying the following conditions:
\begin{enumerate}[{(MN}1{)}]
 \item\label{isotony} {\bf Isotony}: If $I_1 \subset I_2$, then $\A(I_1) \subset \A(I_2)$.
 \item\label{locality} {\bf Locality}: If $I_1 \cap I_2 = \emptyset$, then $\A(I_1)$ and $\A(I_2)$ commute.
 \item\label{covariance} {\bf M\"obius covariance}: for $\g \in \psl2r$, $\Ad U(\g)(\A(I)) = \A(\g I)$.
 \item\label{positive} {\bf Positive energy}: the generator $L_0$ of rotations $U(\rho(t)) = e^{itL_0}$ is positive.
 \item\label{vacuum} {\bf Vacuum}: $\Omega$ is the unique (up to a phase) invariant vector for $U(\g), \g\in \psl2r$
 and $\overline{\A(I)\Omega} = \H$ (the {\bf Reeh-Schlieder property}).
 \setcounter{axioms}{\value{enumi}}
\end{enumerate}
From (MN\ref{isotony})--(MN\ref{vacuum}), one can prove the following \cite{GF93, FJ96}:
\begin{enumerate}[{(MN}1{)}]
\setcounter{enumi}{\value{axioms}}
\item\label{haag} {\bf Haag duality}: $\A(I') = \A(I)'$, where $I'$ is the interior of the complement of $I$.
\item\label{additivity} {\bf Additivity}: If $I = \bigcup I_j$, then $\A(I) = \bigvee_j \A(I_j)$.
\end{enumerate}
We consider also the following additional properties:
\begin{itemize}
 \item {\bf Strong additivity}: If $I_1$ and $I_2$ are the intervals obtained from $I$ by removing one point,
 then $\A(I_1)\vee \A(I_2) = \A(I)$.
 \item {\bf Conformal covariance}: $U$ extends to a projective unitary representation of the
 group $\diff_+(S^1)$ of orientation preserving diffeomorphisms 
 and $\Ad U(\g)\A(I) = \A(\g I)$, and $\Ad U(\g)(x) = x$ if $I \cap \supp \g = \emptyset$.
\end{itemize}
Strong additivity does not follow from (MN\ref{isotony})--(MN\ref{vacuum}),
and indeed the Virasoro nets with $c>1$ fail to have strong additivity \cite[Section 4]{BS90}.
If conformal covariance holds, $(\A, U, \Omega)$ is called a {\bf conformal net}.

Concrete examples relevant to this work will be presented in Section \ref{examples}.
\subsection{Ground state representations}\label{groundgeneral}
Let us identify $\RR$ with a dense interval in $S^1$ by the stereographic projection.
By convention, the point of infinity $\infty$ corresponds to $-1 \in S^1 \subset \CC$.
We denote by $\A|_\RR$ the net of algebras $\{\A(I): I \Subset \RR\}$
restricted to finite open intervals in the real line $\RR$.
Translations $\tau(t)$ and dilations $\d(s)$ acts on $\A|_\RR$, and they are elements of $\psl2r$.

$\Ad U(\tau(t))$ acts on $\overline{\bigcup_{I \Subset \RR} \A(I)}^{\|\cdot\|}$ as an automorphism.
A {\bf state $\f$} on $\A|_\RR$ is by definition a state on the $C^*$-algebra $\overline{\bigcup_{I \Subset \RR} \A(I)}^{\|\cdot\|}$.
Let $\f$ be a state on $\A|_\RR$ and
let $\pi_\f$ be the GNS representation with respect to $\f$.
If $\pi_\f$ extends to $\A(\RR_\pm)$ in the $\s$-weak topology, then
we say that $\pi_\f$ is {\bf solitonic}.
If there is a representation $U^\f$ of the translation group $\RR$ such that
$\Ad U^\f(\tau(t))(\pi_\f(x)) = \pi_\f(\Ad U(\tau(t))(x))$,
then we say that $\pi_\f$ is {\bf translation-covariant}.
Furthermore, if the generator of $U^\f$ can be taken positive,
then we say $\pi_\f$ is a {\bf positive-energy representation}.

If a state $\f$ on $\A|_\RR$ is invariant under $\Ad U(\tau(t))$,
then $\pi_\f$ is automatically translation-covariant,
because one can define $U^\f(\tau(t))\pi_\f(x)\Omega_\f = \pi_\f(\Ad U(\tau(t))(x))\Omega_\f$,
where $\Omega^\f$ is the GNS vector for $\f$.
If furthermore $U^\f$ has positive energy,
then we say that it is a {\bf ground state representation},
and the state $\f$ is called the {\bf ground state}.
By the Reeh-Schlieder argument (see e.g.\! \cite[Theorem 1.3.2]{Baumgaertel95}, \cite[Theorem 3.2.1]{Longo08}),
the GNS vector $\Omega_\f$ is cyclic for $\pi_\f(\A(I))$ for any $I\Subset \RR$.

If $\pi_\f$ is a ground state representation with the ground state vector $\Omega_\f$,
then one can consider the translation-covariant net $\{\pi_\f(\A(I))\}$ on intervals in $\RR$ on the GNS representation space $\H_\f$.
Its dual net is defined by $\hat\A_\f(\RR_+ + a) := \left(\bigvee_{I \Subset \RR_+ + a}\pi_\f(\A(I))\right)''$
and $\hat \A_\f((a,b)) = \hat\A_\f(\RR_+ + b) \cap \hat\A_\f(\RR_+ + a)'$.
By \cite[Corollary 1.9]{GLW98}, $\hat\A$ extends uniquely to a strongly additive M\"obius covariant net on $S^1$:
\begin{theorem}[Guido-Longo-Wiesbrock]
 There is a one-to-one correspondence between
 \begin{itemize}
  \item Isomorphism classes of strongly additive M\"obius covariant nets
  \item Isomorphism classes of Borchers triples $(\M, U, \Omega)$,
  ($\M$ is a von Neumann algebra, $U$ is a representation of $\RR$ with positive generator,
  $\Omega$ is cyclic and separating for $\M$ and $U(t)\Omega=\Omega$ such that
  $\Ad U(t)(\M) \subset \M$ for $t\ge 0$) with the property that $\Omega$ is cyclic for $\Ad U(t)(\M)'\cap \M$
 \end{itemize}
\end{theorem}
Even if the given net $\A$ is conformally covariant,
it is not known whether the net $\hat \A_\f$ is conformally covariant.
Although $\hat \A_\f$ is $\diff_\mathrm{c}(\RR)$-covariant in the natural sense,
where $\diff_\mathrm{c}(\RR)$ is the group of diffeomorphisms of $\RR$ with compact support,
we do not have the uniqueness results of $\diff_\mathrm{c}(\RR)$-action (cf.\! \cite[Theorem 5.5]{CW05} for a uniqueness
theorem for $\diff_+(S^1)$-action).
See \cite[Chapter 4]{WeinerThesis} for some attempts to construct $\diff_+(S^1)$-covariance),
and \cite[Section 4.2]{MT18} for some examples in $(1+1)$-dimensions where M\"obius
covariance gets lost by passing to the dual net.

The purpose of this paper is to construct ground state representations of certain nets.
One of the ideas to construct ground state representations is the following:
an {\bf automorphism} $\a$ of the net $\A|_\RR$ is a family $\{\a_I\}_{I\Subset \RR}$ of automorphisms of $\{\A(I)\}$
such that $\a_{\tilde I}|_{\A(I)} = \a_I$ for $I \subset \tilde I$.
$\a$ extends naturally to $\overline{\bigcup_{I \Subset \RR} \A(I)}^{\|\cdot\|}$.
\begin{lemma}\label{lm:composed}
 If $\a$ commutes with $\Ad U(\tau(t))$, then $\omega\circ\a$ is a ground state,
 where $\omega = \<\Omega, \cdot\,\Omega\>$ is the vacuum state.
\end{lemma}
\begin{proof}
 The GNS representation of $\omega\circ\a$ is given by $\pi_\omega\circ\a = \a$,
 where $\pi_\omega$ is the vacuum representation, the identity map.
 As $\a$ commutes with $\Ad U(\tau(t))$, the same representation $U(\tau(t))$ with positive energy
 implements translations.
\end{proof}
Furthermore, combining \cite[Lemma 4.4, Lemma 4.5]{CLTW12-1}, we obtain the following dichotomy:
if $\A|_\RR$ admits one nontrivial ground state, then there must be continuously many ground states.
\begin{proposition}\label{pr:dilated}
 If there is an automorphism $\a$ of $\A|_\RR$ commuting with translations
 and if $\omega\circ\a \neq \a$,
 then $\{\omega\circ\a \circ \Ad U(\d(s))\}_{s\in\RR}$ are mutually different
 and their GNS representations $\{\Ad U(\d(-s))\circ\a\circ \Ad U(\d(s))\}_{s\in\RR}$ are mutually unitarily inequivalent. 
\end{proposition}
Note that $\a$ as in Proposition \ref{pr:dilated} cannot be implemented by a unitary operator,
because $\Omega$ is unique up to a phase. Hence, the second statement
implies that the dilation automorphisms $\{\Ad U(\d(s))\}$ are not unitarily implemented in the GNS representations
(namely, there are no $U^\a(s)$ such that $\Ad U^\a(s)\circ \a = \a \circ \Ad U(\d(s))$):
if they were implemented, then it would imply that $\a$ and $\Ad U(\d(-s))\circ\a\circ\Ad U(\d(s))$ are
unitarily equivalent, which is a contradiction.

\section{Main examples}\label{examples}
Here we introduce examples of conformal nets.
Although our emphasis is on their restriction to $\RR$, they are most conveniently defined on $S^1$.
Our treatment of these models follows that of \cite[Section 4.1]{CLTW12-2}.

\subsection{The \texorpdfstring{$\mathrm{U}(1)$}{U(1)}-current net}\label{uone}
Buchholz, Mack and Todorov obtained some fundamental results in the operator-algebraic treatment 
of the $\uone$-current algebra \cite{BMT88}.
The current $J(f)$ can be defined by the Fourier modes $\{J_n\}_{n\in \ZZ}$,
where $J_n$ satisfy the following commutation relations $[J_m, J_n] = m\delta_{m+n,0}$
and hence the current $J$ as a quantum field on $S^1$ satisfies the following relation:
$ [J(f), J(g)] = i\int_{S^1} f(\theta)g'(\theta)d\theta$, for $f, g \in C^\infty(S^1, \RR)$.
It admits a vacuum state $\<\Omega, \cdot\, \Omega\>$ such that $J_n\Omega = 0$ for $n \ge 0$,
and the representation is uniquely determined by this property, unitarity $J_n^* = J_{-n}$
and the commutation relation.
The smeared current $J(f) = \sum_n \hat f_n J_n$, where $f\in C^\infty(S^1,\RR)$ and
$\hat f_n = \frac1{2\pi}\int f(\theta)e^{-in\theta}$ is the Fourier-coefficient,
is essentially self-adjoint
on the subspace $\hfin$ spanned by vectors $J_{-n_1}\cdots J_{-n_k}\Omega$, $n_j \in \NN$.

The vacuum state is invariant under $\psl2r$-transformations
$J(f)\mapsto J(f\circ \g^{-1}), \g\in\psl2r$, hence this can be implemented by a unitary operator $U(\g)$,
and in this way we obtain a unitary representation $U$ of $\psl2r$
with positive energy.

The exponential $W(f) = e^{iJ(f)}$ is called the Weyl operator.
They satisfy $W(f)W(g) = e^{-\frac i2\s(f,g)}W(f+g)$,
where $\s(f,g) = \int_{S^1} f(\theta)g'(\theta)d\theta$.
If we define von Neumann algebras by $\A_\uone(I) = \{W(f): \supp f \subset I\}''$,
then $(\A_\uone, U, \Omega)$ is a M\"obius covariant net.
This is called the {\bf $\uone$-current net}.
It is known that it satisfies strong additivity \cite[(4.21) and below]{BS90}
and conformal covariance \cite[Section 5.3]{Ottesen95}.

We can also define $J$ in the real line picture.
To do this, note that the above commutation relation is invariant under diffeomorphisms.
Therefore, if we introduce $J(f)$ for $f \in C^\infty_\mathrm{c}(\RR,\RR)$
by $J(f\circ C^{-1})$ where $C(t) = -\frac{t-i}{t+i}$ is the Cayley transform mapping $\RR \to S^1$
(with a slight abuse of notation: the definition of $J(f)$ depends on whether $f \in C^\infty_\mathrm{c}(\RR)$
or $f\in C^\infty(S^1)$), we obtain the same commutation relation
\begin{align}\label{eq:u1comm}
 [J(f), J(g)] = i\int_\RR f(t)g'(t)dt, \qquad f, g \in C^\infty(\RR, \RR).
\end{align}
We call it $J$ in the real line picture. Again by the diffeomorphism covariance,
translations and dilations acts on $J$ in the natural way: if $\g$ is such a transformation,
then $\Ad U(\g)(J(f)) = J(f\circ \g^{-1})$.

\begin{lemma}\label{lm:nonsmooth}
 For a function $f$ on $S^1$ such that $\sum_{k\ge 0} k|\hat f_k|^2 < \infty$,
 $J(f)$ and $W(f)$ can be defined by continuity.
\end{lemma}
\begin{proof}
  Note first that if $J(f_n)\Omega\to J(f)\Omega$, then $W(f_n)\to W(f)$ in the strong operator topology.
 Indeed, recall that $\<\Omega, W(g)\Omega\> = e^{-\frac12\|J(g)\Omega\|^2}$ \cite[Section 6.5]{Longo08}
 and $\s(f,g) = 2\im\<J(f)\Omega, J(g)\Omega\>$.
 Therefore, it is straightforward to see that, if $J(f_n)\Omega \to J(f)\Omega$, then $W(f_n)\Omega\to W(f)\Omega$
 and $\s(f_n,g) \to \s(f,g)$.
 It then follows that for a fixed $g$
 \begin{align*}
  \|(W(f_n)-W(f))W(g)\Omega\| \to 0,
 \end{align*}
 and since $W(g)\Omega$ is total and $\{W(f_n)\}$ are unitary, this implies the claimed strong convergence.
 
 For a smooth function $f$ on $S^1$, $\|J(f)\Omega\|^2 = \sum_{k\ge 0} k|\hat f_k|^2$,
 and hence $J$ can be extended by continuity to any $L^2$-function $f$ on $S^1$
 such that $\sum_{k\ge 0} k|\hat f_k|^2 < \infty$.
\end{proof}
If $f$ is piecewise smooth and continuous, then $\sum_{k\ge 0} k|\hat f_k|^2 < \infty$.
Indeed, as an operator on $L^2(S^1)$, $i\frac{d}{d\theta}$ is self-adjoint
on the  following domain \cite[X.1 Example 1]{RSII}
\[
 \{f \in C(S^1): f \text{ is absolutely continuous}, f' \text{ is in } L^2(S^1)\}
\]
and $\sum_k |\hat f'_k|^2 = \sum_k k^2|\hat f_k|^2 < \infty$, therefore, $\sum_{k\ge 0} k|\hat f_k|^2 < \infty$.
These facts can be seen as an easier version of \cite[Proposition 4.5]{CW05}\cite[Lemma 2.2]{Weiner06}.

\subsection{The Virasoro nets}\label{virasoro}
On the Hilbert space of the $\uone$-current net, one can construct another quantum field.
Let $L_n = \frac12 \sum_{m} :J_{-m}J_{n+m}:$ where the normal ordering means $J_k$
with negative $k$ comes to the left \cite[(2.9)]{KR87}.
Then they satisfy the Virasoro algebra $[L_m, L_n] = (m-n)L_{m+n} + c\frac{m(m^2-1)}{12}$
with the central charge $c=1$.
The stress-energy tensor $T(f) = \sum \hat f_n L_n$ is essentially self-adjoint on $\hfin$
and satisfies the commutation relations
$[T(f), J(g)] = i J(fg')$. Especially, if $f$ and $g$ have disjoint support, they commute.
By the linear energy bound \cite[Proposition 4.5]{CW05}\cite[Lemma 3.2]{BT13} and
the Driessler-Fr\"ohlich theorem \cite[Theorem 3.1]{DF77}(or its adaptation to the case of operators \cite[Theorem 3.4]{BT13}),
they actually strongly commute.
For $\g \in \diff(S^1)$, it holds that $\Ad U(\g)(T(f)) = T(\g_*f) + \b(\g,f)$ with a certain $\b(\g,f)\in\RR$,
where $\g_*f$ is the pullback of $f$ as a vector field by $\g$ \cite[Proposition 3.1]{FH05}.
Moreover, $T(f)$ integrates to the representation $U$ in the sense
that $U(\Exp(tf)) = e^{itT(f)}$ up to a phase \cite[Section 4]{GW85}, see also \cite[Proposition 5.1]{FH05}.

Let us denote $\B_1(I) = \{e^{iT(f)}:\supp f\subset I\}''$
and call it the {\bf Virasoro subnet}.
By Haag duality, it holds that $\B_1(I) \subset \A_\uone(I)$.
Furthermore, by the Reeh-Schlieder argument, $\H_{\Vir_1} := \overline{\B_1(I)\Omega}$ does not depend on $I$.
The restriction of $\B_1(I)$ and  $U$ to $\H_{\Vir_1}$ together with $\Omega \in \H_{\Vir_1}$
is called the {\bf Virasoro net with $c=1$} and we denote it by $(\vir_1, U_1, \Omega_1)$.

We can also consider the real line picture for $T$.
For a vector field $f$ on $\RR$, or more precisely $f(t)\frac d{dt}$,
we define $T(f) := T(C^{-1}_*f)$.
Conversely, a vector field $f$ on $S^1$ corresponds to $C_*f(t) = \frac{t^2+1}2 f(C^{-1}(t))$.

In \cite[Proof of Theorem 4.7]{CLTW12-2}, we also computed the exponentiation of the relation between $J(g)$ and $T(f)$:
\begin{align*}
 \Ad W(g)(T(f)) = T(f) - J(fg') + \frac12\s(fg', g),
\end{align*}
where this equality holds on $\hfin$.

In order to obtain the Virasoro nets with $c > 1$, we need to perturb $T$ as in \cite[(4.6)]{BS90}.
We have computed its real-line picture \cite[Section 5.3]{CLTW12-2}:
\begin{align*}
 T^\k(f) = T(f) + \k J(f'),
\end{align*}
then $T^\k$ satisfies the commutation relation
\[
 [T^\k(f), T^\k(g)] = iT^\k(fg'-f'g) + i \frac {1+\k^2}{12}\int_\RR f'''(t)g(t)dt,
\]
hence the central charge is $c = 1+\k^2$.
Define $\B_c(I) := \{e^{iT^\k(f)}: \supp f\subset I\}''$.
We have seen that $\B_c(I) \subset \A_{\uone}(I)$, but $\B_c(I)$ is covariant
only with respect to $U$ restricted to translations and dilations.
Yet, $\H_{\Vir_c} := \overline{\B_c(I)\Omega}$ does not depend on $I$ again
by the Reeh-Schlieder argument. It was shown in \cite[Section 4, (4.8)]{BS90} that
by exploiting the M\"obius invariance of $n$-points functions of $T^\k$,
the restriction of $\B_c$ to $\H_{\vir_c}$ can be extended to a conformal net,
the Virasoro net $\vir_c$ with central charge $c = 1+\k^2$.
$\vir_c, c>1$ does not satisfy strong additivity \cite[Section 4, (4.13)]{BS90}.
It is an open problem whether the dual net of $\vir_c$ is conformal ($\diff(S^1)$-covariant).

\section{Ground states}\label{groundstates}
\subsection{On the \texorpdfstring{$\uone$}{U(1)}-current net}
Following Lemma \ref{lm:composed}, we construct automorphisms of $\A_{\uone}|_\RR$ commuting with translations.
This has been done in \cite[Proposition 4.1]{CLTW12-2}. Let us recall it in our present notations.
\begin{lemma}\label{lm:u1auto}
 For $q \in \RR$, there is an automorphism $\a_q$ of $\A_{\uone}|_\RR$ commuting with translations
 such that $\a_q(W(f)) = e^{iq\int_\RR f(t)dt}W(f)$.
\end{lemma}
\begin{proof}
 Let $I \Subset \RR$. There is a smooth function $s_I$ with compact support such that $s_I(t) = t$.
 By commutation relation \eqref{eq:u1comm}, for $f$ with $\supp f \subset I$
 we have
 \[
 \Ad W(qs_I)(W(f)) = e^{iq\int_\RR s_I'(t)f(t)dt}W(f) = e^{iq\int_\RR f(t)dt}W(f).
 \]
 We define $\a_{q,I} = \Ad W(qs_I)$. This does not depend on the choice of $s_I$ within the restriction above,
 because $\A_{\uone}(I)$ is generated by $W(f)$ with $\supp f \subset I$.
 $\a_q$ commutes with translations because $\int_\RR f(t)dt$ is invariant under translations.
\end{proof}

From here, it is immediate to construct ground states by Lemma \ref{lm:composed}.
We collect some properties of the resulting GNS representations.
\begin{theorem}\label{th:u1}
 We have the following.
 \begin{itemize}
  \item  The states $\{\omega \circ \a_q\}$ are mutually different ground states of $\A_{\uone}|_\RR$,
  and they are connected with each other by dilations.
  Dilations are not implemented unitarily.
  \item 
  The GNS vector $\Omega_q$ of $\omega \circ \a_q$ is in the domain of $J_q(f)$,
  where $\a_q(W(f)) = e^{iJ_q(f)}$ and $J_q(f) = J(f) + q\int_\RR f(t)dt$.
  Especially, it holds that $\omega \circ \a_q(J(f)) = q\int_\RR f(t)dt$
  \item The GNS representations $\{\a_q\}$ do not extend to $\A(\RR_\pm)$ in the weak operator topology,
  hence are not solitonic.
  \item The dual nets $\{\hat \a_q(\A_{\uone}(I))\}$ are equal to the original net $\A_{\uone}$.
 \end{itemize}
\end{theorem}
\begin{proof}
 As $\{\a_q\}$ are automorphisms commuting with translations,
 the states $\{\omega \circ \a_q\}$ are ground states by Lemma \ref{lm:composed},
 and their GNS representations are $\{\a_q\}$.
 They are different states because they give the scalar $e^{iq\int_\RR f(t)dt}$ to $W(f)$.
 It follows from Proposition \ref{pr:dilated} that
 they are connected by dilations and dilations are not implemented unitarily.
 
 Let us recall the action $\a_q(W(sf)) = e^{iqs\int_\RR f(t)dt}W(sf)$.
 From this it is immediate that
 \[
  \frac{d}{ds}\a_q(W(sf))\Omega|_{s=0} = iJ(f)\Omega + iq\int_\RR f(t)dt\cdot \Omega.
 \]
 We also showed \cite[Lemma 4.6]{CLTW12-2} that
 $\H^\infty = \bigcap_n \dom(L_0^n)$ is invariant under Weyl operators $W(g)$,
 hence we have, the following equation of operators on $\H^\infty$:
 \[
  \Ad W(q\s_I)(J(f)) = J(f) + q\int_\RR f(t)dt.
 \]
 $\H^\infty$ is a core of $L_0$, and hence of $J(f)$
 by the commutator theorem \cite[Theorem X.37]{RSII} and the estimate
 $\|J(f)\xi\| \le c_f\|(L_0+1)^\frac12 \xi\|$, see the computations in \cite[Proposition 1]{Lechner03}
 (this estimate is known since \cite[below (2.23)]{BS90}.
 Since $\H^\infty$ is invariant under $J(f)$, to apply the commutator theorem with slightly different assumptions \cite[Theorem X.36]{RSII},
 it is enough to have a linear bound $\|J(f)\xi\| \le c_f\|(L_0+1) \xi\|$).

 To show the non-normality on the half line $\RR_+$ ($\RR_-$ is analogous),
 we take a sequence of smooth functions $g_n(-e^{i\theta})$ on $S^1$ (see Figure \ref{fig:g}):
 \[
  g_n(-e^{i\theta}) = \left\{\begin{array}{ll}
                      0 & \text{ for } 0 \le \theta < \pi  \\
                      \theta - \pi & \text{ for } \pi \le \theta < \frac{3\pi}2 \\
                      -\theta + 2\pi & \text{ for } \frac{3\pi}2 \le \theta < 2\pi - \frac1{2n} \\
                      - 2(\theta - 2\pi + \frac1{2n}) & \text{ for } 2\pi - \frac1n \le \theta < 2\pi - \frac1{2n} \\
                      0 & \text{ for } 2\pi - \frac1{2n}\le \theta \le 2\pi
                     \end{array}\right.
 \]
\begin{figure}[ht]\centering
\begin{tikzpicture}[scale=1.2]
        \draw [->] (-1,0) --(5,0) node [above right] {$\theta$};
         \draw [->] (0,-1)--(0,2) node [ right] {$g(-e^{i\theta})$};
          \draw [ultra thick] (0,0)-- (2,0)-- (3,1)--(3.6,0.4)-- (3.8,0)-- (4,0) ;
          \draw [dotted] (0,1)-- (3,1) ;
          \draw [thick](3.6,0.4)-- (4,0) ;
          \node at(0.2,-0.3) {$0$};
          \node at(4,-0.3) {$2\pi$};
          \node at(2,-0.3) {$\pi$};
          \node at(-0.3,1) {$\frac\pi 2$};
\end{tikzpicture}
\caption{The functions $g_n(-e^{i\theta})$ (the thick line) and $g(-e^{i\theta})$ (the thin line above),
restricted to $0 \le \theta \le 2\pi$.}
\label{fig:g}
\end{figure}
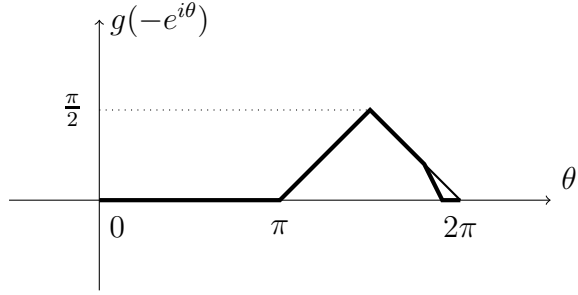
 It is clear that each $g_n$ is piecewise smooth (linear) and continuous,
 and both $g_n, g_n'$ converges in the $L^2$-norm to
 \[
  g(-e^{i\theta}) = \left\{\begin{array}{ll}
                      0 & \text{ for } 0 \le \theta < \pi  \\
                      \theta - \pi & \text{ for } \pi \le \theta < \frac{3\pi}2 \\
                      -\theta + 2\pi & \text{ for } \frac{3\pi}2 \le \theta < 2\pi \\
                     \end{array}\right..
 \]
 Therefore, $J(g_n)\Omega \to J(g)$ and hence $W(g_n) \to W(g)$ in the strong operator topology, and in the weak operator topology, too.
 On the other hand,
 $\a_q(W(g_n)) = e^{iq_n}W(g_n)$, where
 \[
  q_n = q \int_\RR g_n(C(t))dt.
 \]
 Note that, $C(t) = -\frac{t-i}{t+i}$, and $g(C(t))$ behaves as $\frac 1t$ as $t \to \infty$,
 hence $q_n \to \infty$. Accordingly, $\a_q(W(g_n)) = e^{iq_n}W(g_n)$
 is not convergent.
 
 The last claim follows immediately because $\a_q$ are automorphisms, hence
 $\a_q(\A_\uone(I)) = \A_\uone(I)$ and by strong additivity of $\A_\uone$. 
\end{proof}
In this case it is trivial that the dual net coincides with the original net.
In contrast, we have a markedly different result for $\vir_c$.

The fact that $\omega\circ \a_q(J(f)) = q\int_\RR f(t)dt$ 
suggests that $q$ is an analogue of the lowest weight $g$ of \cite[Section 1C]{BMT88}, cf.\! \cite[Lemmas 5.3, 5.6]{Tanimoto11}
and the state $\Omega_q$ has uniform charge density $q$.

\paragraph{The $\uone$-current as a multiplier representation of the loop group $SLS^1$.}
Let $SLS^1$ be the group of smooth maps from $S^1$ into $S^1$ with winding number $0$ with pointwise multiplication \cite[Section 5.1]{Ottesen95}.
These elements can be identified with smooth functions on $S^1$ in $\RR$ by logarithm,
and the group operation becomes addition.
Recall the Weyl relation $W(f)W(g) = e^{-\frac i 2 \s(f,g)}W(f+g)$:
this can be seen as a multiplier representation of $SLS^1$ with the cocycle $e^{-\frac i 2 \s(f,g)}$.
From this perspective, one can also consider the subgroup $S\Omega S^1$
of functions $f$ on $S^1$ such that $f^{(n)}(-1) = 0$ for all $n\in\NN$
\cite[Section 2.1.2]{Henriques17}.

Let us observe that $f\mapsto \a_q(W(f))$ is a multiplier representation
of functions $f$ whose support does not contain $-1$,
because for such $f$, $\a_q(W(f))$ is defined and respects product.
We claim that this extends to $f\in S\Omega S^1$.
Indeed, for such $f$, $f(C(t))$ is a rapidly decreasing smooth function, hence
we can find compactly supported test functions $f_n \to f$ in the topology of Schwartz-class functions,
and then $\a_q(W(f_n)) = e^{iq\int_\RR f_n(t)dt}W(f_n) \to e^{iq\int_\RR f(t)dt}W(f)$.
On the other hand, we have seen in Theorem \ref{th:u1} that
the representation $\a_q$ is not normal (continuous in the $\s$-weak topology, which coincides with
the weak operator topology on norm-bounded sets) on any half-line.

Henriques conjectured that any positive-energy representation of $\Omega G$
should be normal on half-lines \cite[Conjecture 34]{Henriques17}, where $G$ is a simple and simply connected compact group.
The above observation shows that an analogous conjecture for $S\Omega S^1$ does not hold.
On the other hand, we have also shown that for such $G$, the only ground state of the algebra $\mathscr{S}\mathfrak{g}_\CC$
of Schwartz-class maps in $\mathfrak{g}_\CC$ is the vacuum state \cite[Corollary 5.8]{Tanimoto11}.
This supports to some extent the conjecture for $\Omega G$.
Yet, there are many other positive-energy representations of $\Omega G$ without ground state,
and the situation is not conclusive.


\subsection{On Virasoro nets}
We can simply restrict the states $\{\omega \circ \a_q\}$ to the Virasoro nets $\vir_c$.
Let us denote $\H^\k_q = \overline{\a_q(\vir_c(I))\Omega}$, where $c = 1 + \k^2$
and this does not depend on $I$ by the Reeh-Schlieder argument.
Let us denote the GNS representation $P_{\H^\k_q}\a_q$ by $\r_{\k,q}$,
where $P_{\H^\k_q}$ is the projection onto $\H^\k_q$.

\begin{theorem}
 We have the following.
 \begin{itemize}
  \item  The states $\{\omega \circ \a_q\}$ are mutually different ground states of $\vir_c|_\RR$
  for different $\frac{q^2}2$,
  and they are connected with each other by dilations.
  Dilations are not implemented unitarily in $\r_{\k,q}$.
  \item 
  The GNS vector $\Omega_q$ of $\omega \circ \a_q$ is in the domain of $T^\k_q(f)$,
  where $\a_q(e^{isT^\k(f)}) = e^{iT^\k_q(f)}$ and $T^\k_q(f) = T(f) + qJ(f) + \k J(f') + \frac{q^2}2\int_\RR f(t)dt$
  on $\H^\infty$.
  \item If $c > 1$, the dual nets $\{\hat \r_q(\vir_c(I))\}$ are not unitarily equivalent to any of $\vir_{c'}$, $c' >  1$.
 \end{itemize}
\end{theorem}
\begin{proof}
 Let us first prove the second statement.
 We have seen in \cite[Lemma 4.6]{CLTW12-2}
 that $W(g)$ preserves $\H^\infty$ and for $I \supset \supp f$ it holds on $\H^\infty$ that
 \[
  \Ad W(qs_I)(T(f)) = T(f) + q J(f) + \frac{q^2}2\int_\RR f(t)dt.
 \]
 Furthermore, we have $\Ad W(qs_I)(J(f)) = J(f) + q\int_\RR f(t)dt$,
 and $T^\k(f) = T(f) + \k J(f')$.
 Applying $\Ad W(qs_I)$ term by term, we obtain on $\H^\infty$
 \[
  \Ad W(qs_I)(T^\k(f)) = T(f) + q J(f) + \k J(f')+ \frac{q^2}2\int_\RR f(t)dt,
 \]
 since $\int_\RR f'(t) dt = 0$, and this operator has $\Omega$ in its domain.
 
 To show that they are different states for different $\frac{q^2}2$,
 it is enough to observe that
\[
 \frac{d}{ds}\omega\circ\a_q(e^{isT^\k(f)})
 = \<\Omega, \Ad W(qs_I)(T^\k(f))\Omega\> = \frac{q^2}2\int_\RR f(t)dt,
\]
 giving different values.
 
 Let us also check that these states depend only on $c = 1+\k^2$ and not on $\k$.
 To see this, note that $J(f)\mapsto J(-f)$ is a vacuum-preserving automorphism of $\A|_\RR$
 commuting with translations. This automorphism $\a_-$clearly intertwines $\a_q$ and $\a_{-q}$.
 Under $\a_-$, $T(f)$ is mapped to $T(f)$ itself because it is quadratic in $J$,
 therefore, $T^\k(f)$ is mapped to $T^{-\k}(f)$.
 Since $\omega \circ \a_q$ and $\omega \circ \a_q \circ \a_- = \omega \circ \a_- \circ \a_{-q} = \omega \circ \a_{-q}$ give the same state on $\vir_c$,
 they depend only on $c = 1+\k^2$.
 
 They are connected by dilation $\d(s)$ on the larger algebra $\A_\uone|_\RR$
 and from the action\footnote{Formally, $J$ has scaling dimension $1$:
 $\Ad (U(\d(s))(J(t)) = e^{s}J(e^s t)$.
 By integrating this equation again $f$, we obtain $\Ad U(\d(s))(J(f)) = J(f\circ \d(-s))$,
 where $f\circ \d(-s)(t) = f(e^{-s}t)$.
 Similarly, the stress-energy tensor $T$
 has scaling dimension $2$: $\Ad (U(\d(s))(T(t)) = e^{2s}T(e^s t)$
 and $\Ad U(\d(s))(T(f)) = e^sT(f\circ \d(-s))$.} of dilations $\Ad U(\d(s)) (J(f)) = J(f\circ \d(-s))$
 it follows that
 \begin{align*}
  \Ad [W(qs_I)U(\d(s))] (J(f)) &= J(f\circ \d(-s)) + q\int_\RR f\circ\d(-s)(t)dt \\
  &= J(f\circ \d(-s)) + e^s q\int_\RR f(t)dt \\
  &= \Ad U(\d(s))(J_{e^s q}(f))
 \end{align*}
 Hence $\omega\circ\a_q \circ \Ad U(\d(s))= \omega \circ \a_{e^s q}$,
 
 Since $\Omega_q$ is the unique translation-invariant vector,
 it follows that $\r_{\k,q}$ are irreducible (see the arguments of \cite[Corollary 1.2.3]{Baumgaertel95}).
 If dilations were implemented by unitary operators,
 they would bring a vector state to another vector state which are again translation-invariant,
 which contradicts the uniqueness of invariant vector.
 
 As $\vir_c, c>1$ are not strongly additive while dual nets are strong additive,
 they cannot be unitarily equivalent.
\end{proof}
We expect that the GNS representations $\{\r_{\k,q}\}$ are not normal on $\bigcup_{I \Subset \RR_\pm}\vir_c(I)$.
A candidate for a convergent sequence which is not convergent in $\{\r_{k,q}\}$
is $e^{T(f_n)}$, where $f_n$ are smooth functions on $\RR$, supported in $\RR_+$ and
converging to a function $f$ which decays as $\frac1t$.
We expect that $T(f) + \k J(f') + qJ(f)$ is self-adjoint, while
$q^2\int_\RR f(t)dt$ is divergent. To do this, the commutator theorem would not suffice
(because formally $[L_0, J(f)] = J(f')$ would not be defined for this $f$)
but one would need an analysis similar to \cite[Theorem 4.4]{CW05}

Here the value $\frac{q^2}2$ can be viewed as the energy density of the state $\Omega_q$,
hence the representation as a whole has infinite energy with respect the original vacuum,
cf.\! \cite{Dybalski08} where it was shown that, under certain conditions, translation-invariant
states cannot be created by finite energy.
This also supports that $\r_{\k,q}$ are not normal even on half-lines.

Even for the case $c=1$,
the dual nets $\{\hat \r_q(\vir_1(I))\}$ cannot be easily identified with any known conformal net.
We remark that strong additivity of the original net $\vir_1$
does not appear to imply that $\{\r_q(\vir_c(I))\}$ is already dual,
cf.\! \cite[Lemma 3.2]{CLTW12-1}.

\section{Outlook}\label{outlook}
\paragraph{Classification.}
As we have seen, the ground states we constructed in this paper are characterized
by a number $q$ (for the $\uone$-current net) or $\frac{q^2}2$ (for the Virasoro nets). This is a structure
very similar to the invariant $\psi$ of \cite[Lemma 5.3]{Tanimoto11}.
If one studies corresponding Lie algebras, it should be possible to classify
ground state representations as in \cite[Theorem 5.6]{Tanimoto11}.

More precisely, to the net $\A_\uone|_\RR$ there corresponds
the central extension of the Lie algebra of compactly supported smooth functions on $\RR$
with the relation $[f,g] = \s(f,g)$. To the Virasoro nets $\vir_c|_\RR$
one considers the central extension of smooth vector fields on $\RR$
with $[f,g] = fg'-f'g + \frac{c}{12}\int f(t)g'''(f)dt$.
By positivity of energy, one should be able to determine the representation
in terms of $q$ and $\frac{q^2}2$, respectively.

Actually, for the $\uone$-current net, the direct classification of ground states
might be possible. Some techniques from \cite[Section 4.2]{CLTW12-2} used to classify KMS states should be useful,
although the boundary condition ($t\to t+i\beta$ for $\beta$-KMS states) is missing for ground states.
As for the Virasoro nets, it is still not clear whether such classification results pass to
the Virasoro nets. For that implication, one needs that the ground state vector
is infinitely differentiable. We are currently not able to do this, because
the corresponding Lie group $\diff_\mathrm{c}(\RR)$ of compactly supported diffeomorphisms of $\RR$
do not possess any compact subgroup (in contrast to $\diff(S^1)$ which contains
finite covers of $\psl2r$, which is the key to differentiate representations \cite[Appendix]{Carpi04},
cf.\! \cite{Zellner17}).

\paragraph{Dual nets.}
The most fundamental properties of the dual nets $\{\hat\rho_{\k,q}(\Vir_c(I))\}$ remain open.
Among them is conformal covariance. Conformal covariance implies the split property \cite{MTW18},
and even the split property is unknown to hold in these dual nets.
The split property may fail in the dual net in two-dimensional Haag-Kastler net \cite[Section 4.2]{MT18},
therefore, it may be worthwhile to try to (dis)prove the split property in these nets.

\paragraph{More positive-energy representations.}
Ground states consist only a particular class of positive-energy representations.
It was shown in \cite{DS82, DS83} that there is a huge class of locally normal positive-energy representations
of the free massless fermion field in $(3+1)$-dimensions. A similar construction should be possible
in one dimension. Furthermore, important conformal nets, including some loop group nets, can be realized as subnets of
(the tensor product of ) the free fermion field nets (see \cite[Examples 4.13-16]{Tener16}).
Among them, there might be counterexamples to \cite[Conjecture 34]{Henriques17}.

\subsubsection*{Acknowledgements.}
We acknowledge the MIUR Excellence Department Project awarded
to the Department of Mathematics, University of Rome Tor Vergata, CUP E83C18000100006.

{\small
\def\cprime{$'$} \def\polhk#1{\setbox0=\hbox{#1}{\ooalign{\hidewidth
  \lower1.5ex\hbox{`}\hidewidth\crcr\unhbox0}}} \def\cprime{$'$}

}

\begin{thebibliography}{CLTW12b}

\bibitem[AKTH77]{AHKT77}
Huzihiro Araki, Daniel Kastler, Masamichi Takesaki, and Rudolf Haag.
\newblock Extension of {KMS} states and chemical potential.
\newblock {\em Comm. Math. Phys.}, 53(2):97--134, 1977.
\newblock \url{https://projecteuclid.org/euclid.cmp/1103900637}.

\bibitem[Bau95]{Baumgaertel95}
Hellmut Baumg{\"a}rtel.
\newblock {\em Operator algebraic methods in quantum field theory}.
\newblock Akademie Verlag, Berlin, 1995.
\newblock \url{https://books.google.com/books?id=rq3vAAAAMAAJ}.

\bibitem[BDLR92]{BDLR92}
Detlev Buchholz, Sergio Doplicher, Roberto Longo, and John~E. Roberts.
\newblock A new look at {G}oldstone's theorem.
\newblock {\em Rev. Math. Phys.}, (Special Issue):49--83, 1992.
\newblock \url{https://www.researchgate.net/publication/252355725}.

\bibitem[BMT88]{BMT88}
Detlev Buchholz, Gerhard Mack, and Ivan Todorov.
\newblock The current algebra on the circle as a germ of local field theories.
\newblock {\em Nuclear Phys. B Proc. Suppl.}, 5B:20--56, 1988.
\newblock \url{https://www.researchgate.net/publication/222585851}.

\bibitem[BR87]{BR1}
Ola Bratteli and Derek~W. Robinson.
\newblock {\em Operator algebras and quantum statistical mechanics. 1. $C^*$-
  and $W^*$-algebras, symmetry groups, decomposition of states}.
\newblock Texts and Monographs in Physics. Springer-Verlag, New York, second
  edition, 1987.
\newblock \url{https://books.google.com/books?id=SV3oCAAAQBAJ}.

\bibitem[BR97]{BR2}
Ola Bratteli and Derek~W. Robinson.
\newblock {\em Operator algebras and quantum statistical mechanics. 2.
  Equilibrium states. Models in quantum statistical mechanics}.
\newblock Texts and Monographs in Physics. Springer-Verlag, Berlin, second
  edition, 1997.
\newblock \url{https://books.google.com/books?id=GSf0CAAAQBAJ}.

\bibitem[BSM90]{BS90}
Detlev Buchholz and Hanns Schulz-Mirbach.
\newblock Haag duality in conformal quantum field theory.
\newblock {\em Rev. Math. Phys.}, 2(1):105--125, 1990.
\newblock \url{https://www.researchgate.net/publication/246352668}.

\bibitem[BT13]{BT13}
Marcel Bischoff and Yoh Tanimoto.
\newblock Construction of {W}edge-{L}ocal {N}ets of {O}bservables through
  {L}ongo-{W}itten {E}ndomorphisms. {II}.
\newblock {\em Comm. Math. Phys.}, 317(3):667--695, 2013.
\newblock \url{https://arxiv.org/abs/1111.1671}.

\bibitem[Car04]{Carpi04}
Sebastiano Carpi.
\newblock On the representation theory of {V}irasoro nets.
\newblock {\em Comm. Math. Phys.}, 244(2):261--284, 2004.
\newblock \url{https://arxiv.org/abs/math/0306425}.

\bibitem[CLTW12a]{CLTW12-1}
Paolo Camassa, Roberto Longo, Yoh Tanimoto, and Mih{\'a}ly Weiner.
\newblock Thermal states in conformal {QFT}. {I}.
\newblock {\em Comm. Math. Phys.}, 309(3):703--735, 2012.
\newblock \url{https://arxiv.org/abs/1101.2865}.

\bibitem[CLTW12b]{CLTW12-2}
Paolo Camassa, Roberto Longo, Yoh Tanimoto, and Mih{\'a}ly Weiner.
\newblock Thermal {S}tates in {C}onformal {QFT}. {II}.
\newblock {\em Comm. Math. Phys.}, 315(3):771--802, 2012.
\newblock \url{https://arxiv.org/abs/1109.2064}.

\bibitem[CW05]{CW05}
Sebastiano Carpi and Mih{\'a}ly Weiner.
\newblock On the uniqueness of diffeomorphism symmetry in conformal field
  theory.
\newblock {\em Comm. Math. Phys.}, 258(1):203--221, 2005.
\newblock \url{https://arxiv.org/abs/math/0407190}.

\bibitem[DF77]{DF77}
W.~Driessler and J.~Fr{\"o}hlich.
\newblock The reconstruction of local observable algebras from the euclidean
  green's functions of relativistic quantum field theory.
\newblock {\em Annales de L'Institut Henri Poincare Section Physique
  Theorique}, 27:221--236, 1977.
\newblock \url{https://eudml.org/doc/75959}.

\bibitem[DHR69]{DHR69}
Sergio Doplicher, Rudolf Haag, and John~E. Roberts.
\newblock Fields, observables and gauge transformations. {I}.
\newblock {\em Comm. Math. Phys.}, 13:1--23, 1969.
\newblock \url{http://projecteuclid.org/euclid.cmp/1103841481}.

\bibitem[DS82]{DS82}
Sergio Doplicher and Mauro Spera.
\newblock Representations obeying the spectrum condition.
\newblock {\em Comm. Math. Phys.}, 84(4):505--513, 1982.
\newblock \url{http://projecteuclid.org/euclid.cmp/1103921286}.

\bibitem[DS83]{DS83}
Sergio Doplicher and Mauro Spera.
\newblock Local normality properties of some infrared representations.
\newblock {\em Comm. Math. Phys.}, 89(1):19--25, 1983.
\newblock \url{http://projecteuclid.org/euclid.cmp/1103922588}.

\bibitem[Dyb08]{Dybalski08}
Wojciech Dybalski.
\newblock A sharpened nuclearity condition and the uniqueness of the vacuum in
  {QFT}.
\newblock {\em Comm. Math. Phys.}, 283(2):523--542, 2008.
\newblock \url{https://arxiv.org/abs/0706.4049}.

\bibitem[FH05]{FH05}
Christopher~J. Fewster and Stefan Hollands.
\newblock Quantum energy inequalities in two-dimensional conformal field
  theory.
\newblock {\em Rev. Math. Phys.}, 17(5):577--612, 2005.
\newblock \url{https://arxiv.org/abs/math-ph/0412028}.

\bibitem[FJ96]{FJ96}
Klaus Fredenhagen and Martin J{\"o}r{\ss}.
\newblock Conformal {H}aag-{K}astler nets, pointlike localized fields and the
  existence of operator product expansions.
\newblock {\em Comm. Math. Phys.}, 176(3):541--554, 1996.
\newblock \url{https://projecteuclid.org/euclid.cmp/1104286114}.

\bibitem[GF93]{GF93}
Fabrizio Gabbiani and J{\"u}rg Fr{\"o}hlich.
\newblock Operator algebras and conformal field theory.
\newblock {\em Comm. Math. Phys.}, 155(3):569--640, 1993.
\newblock \url{http://projecteuclid.org/euclid.cmp/1104253398}.

\bibitem[GLW98]{GLW98}
D.~Guido, R.~Longo, and H.-W. Wiesbrock.
\newblock Extensions of conformal nets and superselection structures.
\newblock {\em Comm. Math. Phys.}, 192(1):217--244, 1998.
\newblock \url{https://arxiv.org/abs/hep-th/9703129}.

\bibitem[GW85]{GW85}
Roe Goodman and Nolan~R. Wallach.
\newblock Projective unitary positive-energy representations of {${\rm
  Diff}(S^1)$}.
\newblock {\em J. Funct. Anal.}, 63(3):299--321, 1985.
\newblock
  \url{http://www.sciencedirect.com/science/article/pii/0022123685900904}.

\bibitem[Haa96]{Haag96}
Rudolf Haag.
\newblock {\em Local quantum physics}.
\newblock Texts and Monographs in Physics. Springer-Verlag, Berlin, second
  edition, 1996.
\newblock \url{https://books.google.com/books?id=OlLmCAAAQBAJ}.

\bibitem[Hen17]{Henriques17}
Andr\'e Henriques.
\newblock Loop groups and diffeomorphism groups of the circle as colimits.
\newblock {\em to appear}, 2017.
\newblock \url{https://arxiv.org/abs/1706.08471}.

\bibitem[KR87]{KR87}
V.~G. Kac and A.~K. Raina.
\newblock {\em Bombay lectures on highest weight representations of
  infinite-dimensional {L}ie algebras}, volume~2 of {\em Advanced Series in
  Mathematical Physics}.
\newblock World Scientific Publishing Co. Inc., Teaneck, NJ, 1987.
\newblock \url{https://books.google.com/books?id=0P23OB84eqUC}.

\bibitem[Lec03]{Lechner03}
Gandalf Lechner.
\newblock Polarization-free quantum fields and interaction.
\newblock {\em Lett. Math. Phys.}, 64(2):137--154, 2003.
\newblock \url{https://arxiv.org/abs/hep-th/0303062}.

\bibitem[Lon08]{Longo08}
Roberto Longo.
\newblock Real {H}ilbert subspaces, modular theory, {${\rm SL}(2,{\bf R})$} and
  {CFT}.
\newblock In {\em Von {N}eumann algebras in {S}ibiu: {C}onference
  {P}roceedings}, pages 33--91. Theta, Bucharest, 2008.
\newblock
  \url{https://www.mat.uniroma2.it/longo/Lecture-Notes_files/LN-Part1.pdf}.

\bibitem[LT18]{LT18}
Roberto Longo and Yoh Tanimoto.
\newblock Rotational {KMS} states and type {I} conformal nets.
\newblock {\em Comm. Math. Phys.}, 357(1):249--266, 2018.
\newblock \url{https://arxiv.org/abs/1608.08903}.

\bibitem[MT18]{MT18}
Vincenzo Morinelli and Yoh Tanimoto.
\newblock Scale and {M}\"obius covariance in two-dimensional haag-kastler net.
\newblock 2018.
\newblock \url{https://arxiv.org/abs/1807.04707}.

\bibitem[MTW18]{MTW18}
Vincenzo Morinelli, Yoh Tanimoto, and Mih\'aly Weiner.
\newblock Conformal covariance and the split property.
\newblock {\em Comm. Math. Phys.}, 357(1):379--406, 2018.
\newblock \url{https://arxiv.org/abs/1609.02196}.

\bibitem[Ott95]{Ottesen95}
Johnny~T. Ottesen.
\newblock {\em Infinite-dimensional groups and algebras in quantum physics},
  volume~27 of {\em Lecture Notes in Physics. New Series m: Monographs}.
\newblock Springer-Verlag, Berlin, 1995.
\newblock \url{https://books.google.com/books?id=7Cn6CAAAQBAJ}.

\bibitem[RS75]{RSII}
Michael Reed and Barry Simon.
\newblock {\em Methods of modern mathematical physics. {II}. {F}ourier
  analysis, self-adjointness}.
\newblock Academic Press, New York, 1975.
\newblock \url{https://books.google.com/books?id=Kz7s7bgVe8gC}.

\bibitem[Tan11]{Tanimoto11}
Yoh Tanimoto.
\newblock Ground state representations of loop algebras.
\newblock {\em Ann. Henri Poincar\'e}, 12(4):805--827, 2011.
\newblock \url{https://arxiv.org/abs/1005.0270}.

\bibitem[Ten16]{Tener16}
James Tener.
\newblock Geometric realization of algebraic conformal field theories.
\newblock 2016.
\newblock \url{https://arxiv.org/abs/1611.01176}.

\bibitem[Wei05]{WeinerThesis}
Mih\'aly Weiner.
\newblock Conformal covariance and related properties of chiral qft.
\newblock 2005.
\newblock Ph.D.\! thesis, Universit\'a di Roma ``Tor Vergata''.
  \url{http://arxiv.org/abs/math/0703336}.

\bibitem[Wei06]{Weiner06}
Mih{\'a}ly Weiner.
\newblock Conformal covariance and positivity of energy in charged sectors.
\newblock {\em Comm. Math. Phys.}, 265(2):493--506, 2006.
\newblock \url{https://arxiv.org/abs/math-ph/0507066}.

\bibitem[Zel17]{Zellner17}
Christoph Zellner.
\newblock On the existence of regular vectors.
\newblock In {\em Representation theory---current trends and perspectives}, EMS
  Ser. Congr. Rep., pages 747--763. Eur. Math. Soc., Z\"urich, 2017.
\newblock \url{https://arxiv.org/abs/1510.08727}.

\end{thebibliography}

\end{document}